\documentclass[a4paper,11pt]{article}
\usepackage{bettera4}

\usepackage{graphics}

\usepackage{url}
\usepackage{diagrams}

\newtheorem{theorem}{\sc Theorem}

\newenvironment{proof}{{\sc Proof.}}{}

\newcommand{\limab}{\lim_{a \rightarrow -\infty, \; b \rightarrow \infty} \;}
\newcommand{\pred}{\mbox{\textrm{pred}}}
\newcommand{\weak}{\mbox{\textit{weak}}}
\newcommand{\vweak}{\mbox{\textit{very weak}}}

\newcommand{\picname}[1]{images/#1.png}
\newarrow{Corresponds}<---> 

\begin{document}

\date{17 June 2013}

\title{When causation does not imply correlation:\\robust violations of the Faithfulness axiom}

\author{Richard Kennaway, jrk@cmp.uea.ac.uk \\
       School of Computing Sciences\\
       University of East Anglia\\
       Norwich NR4 7TJ, UK}

\maketitle

\begin{abstract}
We demonstrate that the Faithfulness property that is assumed in much causal analysis is robustly violated for a large class of systems of a type that occurs throughout the life and social sciences: control systems. These systems exhibit correlations indistinguishable from zero between variables that are strongly causally connected, and can show very high correlations between variables that have no direct causal connection, only a connection via causal links between uncorrelated variables. Their patterns of correlation are robust, in that they remain unchanged when their parameters are varied. The violation of Faithfulness is fundamental to what a control system does: hold some variable constant despite the disturbing influences on it. No method of causal analysis that requires Faithfulness is applicable to such systems.
%
%
\end{abstract}

\section{Introduction}

The problem of deducing causal information from correlations in observational data is a substantial research area, the simple maxim that ``correlation does not imply causation'' having been superceded by methods such as those set out in \cite{Pearl00,GSS01}, and in shorter form in \cite{pearl2009causal}.  The purpose of this paper is to exhibit a substantial class of systems to which these methods and recently developed extensions of them fail to apply.

The works just cited limit attention to systems whose causal connections form a directed acyclic graph, together with certain further restrictions, and also do not consider dynamical systems or time series data. As such, dynamical systems with feedback simply lie outside their scope.
Attempts have been made to extend these methods towards dynamical systems, and systems with cyclic causation, by relaxing one or more of the basic assumptions. However, among the class of dynamical systems with feedback is a subclass of ubiquitous occurrence in the real world, for which, we argue, no such extension of these methods can succeed.

These are control systems: systems which have been designed, whether by a human designer or by evolution, to destroy the correlations that these causal inference methods work from. In addition, they tend to create high correlations between variables that are only indirectly causally linked. In section~\ref{section-literature} we discuss in detail some of the papers that have attempted such extensions. It is not an accident that in every case the assumptions made, while allowing some dynamical systems with feedback, exclude control systems. We show where control systems violate their assumptions, and analyse the result of applying their methods anyway, exhibiting where they break down.

\section{Preliminaries}

\subsection{Causal inference}

We briefly summarise the concepts of causal inference set out in the works cited above.
For full technical definitions the reader should consult the original sources.

A hypothesis about the causal relationships existing among a set of variables $V$ is assumed to be expressible as a directed acyclic graph $G$ (the \textit{DAG} assumption). An arrow from $x$ to $y$ means that there is a direct causal influence of $x$ on $y$, and its absence, that there is none.  Given such a graph, and a joint probability distribution $P$ over $V$, we can ask whether this distribution is consistent with the causal relationships being exactly as stated by $G$: could this distribution arise from these causal relationships? Besides the DAG assumption, there are two further axioms that are generally required for $P$ to be considered consistent with $G$: the \textit{Markov} condition and the \textit{Faithfulness} property.

$P$ satisfies the Markov condition if it factorises as the product of the conditional distributions $P(V_i | \pred(V_i))$, where $\pred(V_i)$ is the set of immediate predecessors of $V_i$ in $G$. This amounts to the assumption that all of the other, unknown influences on each $V_i$ are independent of each other; otherwise put, it is the assumption that $G$ contains all the variables responsible for all of the causal connections that exist among the variables.  It can be summed up as the slogan ``no correlation without causation''.

The Faithfulness assumption is that no conditional correlation among the variables is zero unless it is necessarily so given the Markov property.  For example, if $G$ is a graph of just two nodes $x$ and $y$ with an arrow from $x$ to $y$, then every probability distribution over $x$ and $y$ has the Markov property, but only those yielding a non-zero correlation between $x$ and $y$ are faithful.  It is not obvious in general which of the many conditional correlations for a given graph $G$ must be zero, but a syntactically checkable condition was given by \cite{Pearl98}, called $d$-separation. Its definition will not concern us here.  Faithfulness can be summed up as the slogan ``no causation without correlation''.

The idea behind Faithfulness is that if there are multiple causal connections between $x$ and $y$, then while it is possible that the causal effects might happen to exactly cancel out, leaving no correlation between $x$ and $y$, this is very unlikely to happen.  Technically, if the distributions $P$ are drawn from some reasonable measure space of possible distributions, then the subset of non-faithful distributions has measure zero.

When these assumptions are satisfied, the correlations present in observational data can be used to narrow the set of causal graphs that are consistent with the data.

The assumptions have all been the subject of debate, but we are primarily concerned here with the Faithfulness assumption.
Attacks on it have been based on the argument that very low correlations may be experimentally indistinguishable from zero, and therefore that one may conclude from a set of data that no causal connection can exist even when there is one.  But, it can be countered, that merely reflects on the inadequate statistical power of one's data set, the response to which should be to collect more data rather than question this axiom.  However, we shall not be concerned with this argument.

Instead, our purpose is to exhibit a large class of \textit{robust} counterexamples to Faithfulness: systems which contain zero correlations that do not become nonzero by any small variation of their parameters, nor by the collection of more data, yet are not implied by the Markov property.  Some of these systems even exhibit large correlations (absolute value above 0.95) between variables that have no direct causal connection, but are only connected by a series of direct links, each of which is associated with correlations indistinguishable from zero. These systems are neither exotic, nor artificially contrived for the sole purpose of being counterexamples. On the contrary, systems of this form are common in both living organisms and man-made systems.

It follows that for these systems, this general method of causal analysis of nonexperimental data cannot be applied, however the basic assumptions are weakened. Interventional experiments are capable of obtaining information about the true causal relationships, but for some of these systems it is paradoxically the \textit{lack} of correlation between an experimentally imposed value for $x$ and the observed value of $y$ that will suggest the \textit{presence} of a causal connection between them.

\subsection{Zero correlation between a variable and its derivative}\label{section-zeroderiv}

As a preliminary to the main results in the following sections, we consider the statistical relation between a function, stochastic or deterministic, and its first derivative. In the appendix we demonstrate that under certain mild boundedness conditions, the correlation between a differentiable real function and its first derivative is zero. (The obvious counterexample of $e^x$, identical to its first derivative, violates these conditions.)

An example of a physical system with two variables, both bounded, one being the derivative of the other is that of a voltage source connected across a capacitor. The current $I$ is related to the voltage $V$ by $I = C\, dV/dt$, $C$ being the capacitance.
If $V$ is the output of a laboratory power supply, its magnitude continuously variable by turning a dial, then whatever the word ``causation'' means, it would be perverse to say that the voltage across the capacitor does not cause the current through it. Within the limits of what the power supply can generate and the capacitor can withstand, $I$ can be caused to follow any smooth trajectory by suitably and smoothly varying $V$. The voltage is differentiable, so by Theorem~\ref{theorem-zerocorr-interval} of the Appendix, on any finite interval in which the final voltage is the same as the initial, $c_{V,I}$ is zero. By Theorem~\ref{theorem-zerocorr-realline}, the boundedness of the voltage implies that the same is true in the limit of infinite time.

This is not a merely fortuitous cancelling out of multiple causal connections. There is a single causal connection, the physical mechanism of a capacitor. The mechanism deterministically relates the current and the voltage. (The voltage itself may be generated by a stochastic process.)  Despite this strong physical connection, the correlation between the variables is zero.

Some laboratory power supplies can be set to generate a constant current instead of a constant voltage. When a constant current is applied to a capacitor, the mathematical relation between voltage and current is the same as before, but the causal connection is reversed: the current now causes the voltage.  Within the limits of the apparatus, any smooth trajectory of voltage can be produced by suitably varying the current.

It can be argued that the reason for this paradox is that the product-moment correlation is too insensitive a tool to detect the causal connection.  For example, if the voltage is drawn from a signal generator set to produce a sine wave, a plot of voltage against current will trace a circle or an axis-aligned ellipse.  One can immediately see from such a plot that there is a tight connection between the variables, but one which is invisible to the product-moment correlation.
A more general measure, such as mutual information, would reveal the connection.

However, let us suppose that $V$ is not generated by any periodic source, but varies randomly and smoothly, with a waveform such as that of Figure~\ref{figure-VI}(a). This waveform has been designed to have an autocorrelation time of 1 unit: the correlation between $V(t)$ and $V(t+\delta)$ is zero whenever $|\delta| >= 1$. (It is generated as the convolution of white noise with an infinitely differentiable function which is zero outside a unit interval.\footnote{We used Matlab to perform all the simulations and plot the graphs. The source code is available as supplementary material.})
Choosing the capacitance $C$, which is merely a scaling factor, such that $V$ and $I$ have the same standard deviation, the resulting current is shown in Figure~\ref{figure-VI}(b). A plot of voltage against current is shown in Figure~\ref{figure-VI}(c).  One can clearly see trajectories, but it is not immediately obvious from the plot that there is a simple relation between voltage and current and that no other unobserved variables are involved.  If we then sample the system with a time interval longer than the autocorrelation time of the voltage source, then the result is the scatterplot of Figure~\ref{figure-VI}(d).  The points are connected in sequence, but each step is a random jump whose destination is independent of its source.
Over a longer time, this sampling produces the scatterplot of Figure~\ref{figure-VI}(e). All mutual information between $V$ and $I$ has now been lost: all of the variables $V_i$ and $I_i$ are close to being independently identically distributed. Knowing the exact values of all but one of these variables gives an amount of information about the remaining one that tends to zero as the sampling time step increases.  All measures of any sort of mutual information or causality between them tend to zero, not merely the correlation coefficient. The only way to discover the relationship between $V$ and $I$ is to measure them on timescales short enough to reveal the short-term trajectories instead of the Gaussian cloud.
\begin{figure}
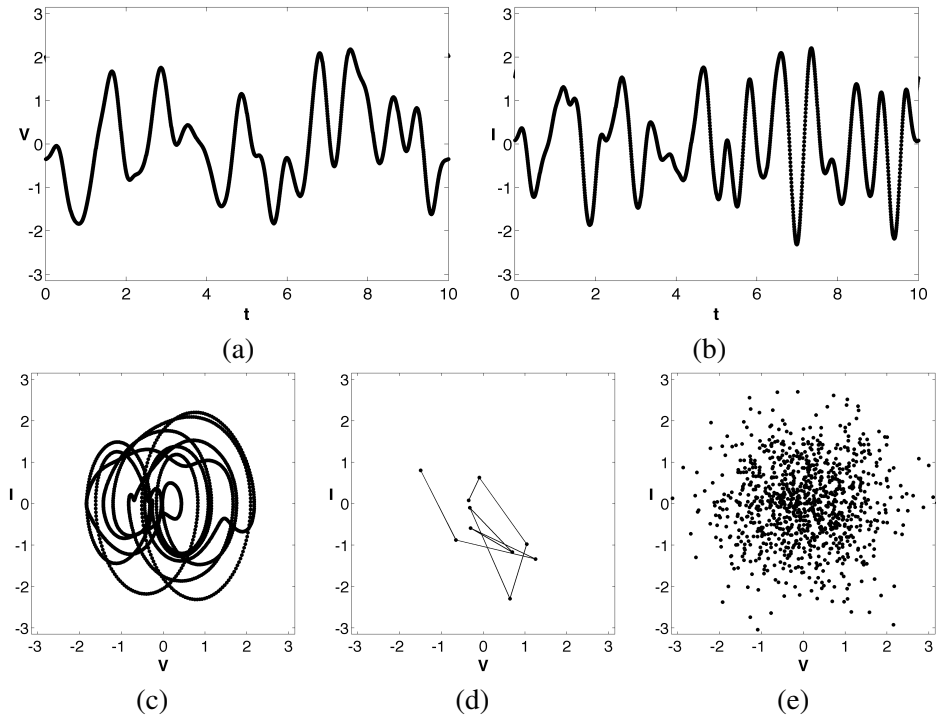

\centering
\begin{tabular}{cc}
\resizebox{!}{4.2cm}{\includegraphics{\picname{fig1pa}}} &
\resizebox{!}{4.2cm}{\includegraphics{\picname{fig1pb}}} \\
(a) & (b) \\
\end{tabular}
\begin{tabular}{ccc}
\resizebox{!}{4.0cm}{\includegraphics{\picname{fig1pc}}} &
\resizebox{!}{4.0cm}{\includegraphics{\picname{fig1pd}}} &
\resizebox{!}{4.0cm}{\includegraphics{\picname{fig1pe}}} \\
(c) & (d) & (e)
\end{tabular}
\caption{Voltage and current related by $I = dV/dt$.
(a) Voltage vs.~time.
(b) Current vs.~time.
(c) Voltage vs.~current.
(d) Voltage vs.~current, sampled.
(e) Voltage vs.~current, sampled for a longer time.
}\label{figure-VI}
\end{figure}

\section{Control systems}

A control system, most generally described, is any device which is able to maintain some measurable property of its environment at or close to some set value, regardless of other influences on that variable that would otherwise tend to change its value. That is a little too general: a nail may serve very well to prevent something from moving, despite the forces applied to it, but we do not consider it to be a control system. Control systems, more usefully demarcated, draw on some energy source to actively maintain the controlled variable at its reference value.
Some everyday examples are a room thermostat that turns heating or cooling mechanisms up and down to maintain the interior at a constant temperature despite variations in external weather, a cruise control maintaining a car at a constant speed despite winds and gradients, and the physiological processes that maintain near-constant deep body temperature in mammals.

The general form of a feedback controller is shown in Figure~\ref{figure-generalcontroller}.
\begin{figure}
\centering
\resizebox{4.86cm}{!}{\includegraphics{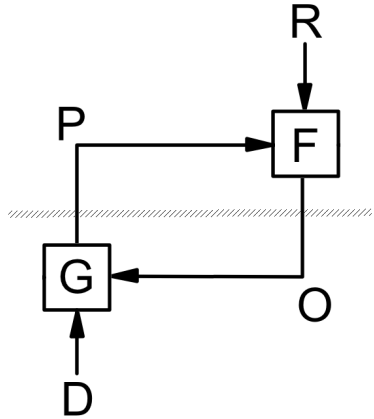}}
\caption{The basic block diagram of any feedback control system. The controller is above the shaded line; its environment (the plant that it controls) is below the line.}\label{figure-generalcontroller}
\end{figure}
\noindent
$\!$The variables have the following meanings:
\begin{description}
\item[$P$:] The controller's perceptual input.  This is a property of the environment, which it is the controller's purpose to hold equal to the reference signal.
\item[$R$:] The reference signal. This is the value that the control system tries to keep $P$ equal to. It is shown as a part of the controller. In an industrial setting it might be a dial set by an operator, or it could be the output of another control system. In a living organism, $R$ will be somewhere inside the organism and may be difficult to discover.
\item[$O$:] The output signal of the controller. This is some function of the perception and the reference (and possibly their past history). This is the action the control system takes to maintain $P$ equal to $R$. Often, and in all the examples of the present paper, $O$ depends only on the difference $R-P$, also called the \textit{error} signal.
\item[$D$:] The disturbance: all of the influences on $P$ besides $O$. $P$ is some function $G$ of the output and the disturbance (and possibly their past history).
\end{description}

We shall now give some very simple didactic examples of control systems, and exhibit the patterns of correlations they yield among $P$, $R$, $O$, and $D$ under various circumstances.

\subsection*{Example 1}

Figure~\ref{figure-integralcontroller} illustrates a simple control system acting within a simple environment, defined by the following equations, all of the variables being time-dependent.
\begin{eqnarray}
\dot{O} & = & k(R - P) \label{eqn-controller}\\
P & = & O + D \label{eqn-environment}
\end{eqnarray}
\begin{figure}
\centering
\resizebox{6cm}{!}{\includegraphics{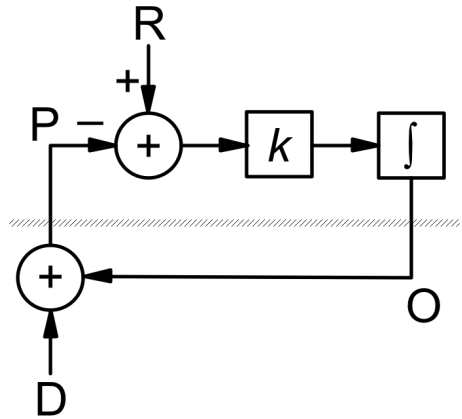}}
\caption{Block diagram of a simple feedback control system.}\label{figure-integralcontroller}
\end{figure}
\noindent
$\!\!$Equation~\ref{eqn-controller} describes an integrating controller, i.e.~one whose output signal $O$ is proportional to the integral of the error signal $R-P$. Equation~\ref{eqn-environment} describes the environment of the controller, which determines the effect that its output action and the disturbing variable $D$ have upon the controlled variable $P$.
In this case $O$ and $D$ add together to produce $P$.
Figure~\ref{figure-controlresponse} illustrates the response to step and random changes in the reference and disturbance. The random changes are smoothly varying with an autocorrelation time of 1 second. The gain $k$ is 100. Observe that when $R$ and $D$ are constant, $P$ converges to $R$ and $O$ to $R-D$.  The settling time for step changes in $R$ or $D$ is of the order of $1/k = 0.01$.
\begin{figure}
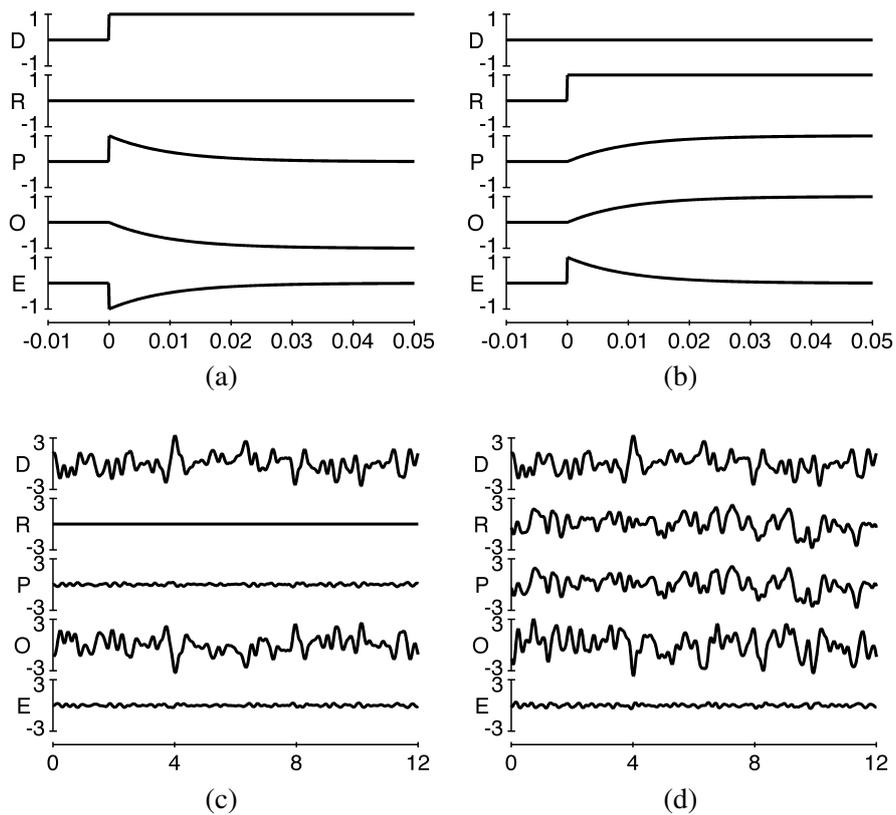

\centering
\begin{tabular}{cc}
\resizebox{!}{4.5cm}{\includegraphics{\picname{fig4pa}}} &
\resizebox{!}{4.5cm}{\includegraphics{\picname{fig4pb}}} \\
(a) & (b) \\
\\
\resizebox{!}{4.5cm}{\includegraphics{\picname{fig4pc}}} &
\resizebox{!}{4.5cm}{\includegraphics{\picname{fig4pd}}} \\
(c) & (d)
\end{tabular}
\caption{Responses of the controller.
(a) Step change in $D$, $R=0$.
(b) Step change in $R$, $D=0$.
(c) $R=0$, randomly varying $D$.
(d) $R$ and $D$ both randomly varying.
}\label{figure-controlresponse}
\end{figure}

The physical connections between $O$, $R$, $P$, and $D$ are as shown in Figure~\ref{figure-integralcontroller}: $O$ is caused by $R$ and $P$, $P$ is caused by $O$ and $D$, and there are no other causal connections. We now demonstrate that the correlations between the variables of this system bear no resemblance to the causal connections.  We generate a smoothly randomly varying disturbance $D$, which varies on a timescale much longer than $1/k$, with a standard deviation of 1 (in arbitrary units).  The reference $R$ is held constant at 0.  Table~\ref{table-ex1-integrating} shows the correlations in a simulation run of 1000 seconds with a time step of 0.001 seconds.
\begin{table}
\[\begin{array}{c|cc}
         &        P &        D \\ \hline
       O &    0.002 &   -0.999 \\
       P &          &   \hfill 0.043 \\
\end{array}\]
\caption{Correlations for an integrating controller (Example 1).}\label{table-ex1-integrating}
\end{table}
The performance of the controller can be measured by its disturbance rejection ratio, $\sigma(D)/\sigma(R-P) = 23.2$.
The numbers vary between runs only in the third decimal place.
For this system, correlations are very high (close to $\pm 1$) exactly where direct causal connection is absent, and close to zero where direct causal connection is present. There is a causal connection between $D$ and $O$, but it proceeds only via $P$, with which neither variable is correlated.

Here is a visual contrast between the causal links and the non-zero pairwise correlations.
\[
\begin{array}{ccc}
\begin{diagram}[size=1.2em,heads=littlevee]
  &      & P \\
  & \ruTo &  & \rdCorresponds \\
D &   &   &   & O
\end{diagram}
& &
\begin{diagram}[size=1.2em]
 &  & P \\
 \\
D & & \hLine^1 & & O
\end{diagram} \\
\mbox{\textit{Causality}} & \hspace{2em} & \mbox{\textit{Non-zero pairwise correlations}}
\end{array}
\]
It is difficult to calculate significance levels for these correlations, since the random waveforms that we generate for $D$ (and in some of our simulations, also for $R$) are deliberately designed to be heavily bandwidth-limited.  They have essentially no energy at wavelengths below about $0.2$ seconds. Successive samples are therefore not independent, and formulas for the standard error of the correlation do not apply. Empirically, if we repeat the simulation many times, we find that the variance of the correlation between two independently generated waveforms like $D$ is proportional to the support interval of the filter that generates them from white noise.  This implies that as one reduces the sampling time step below the point at which there is any new structure to discover, the variance of the correlation will not decrease. When the samples already contain almost all the information there is to find in the signal, increasing the sampling rate cannot yield new information.

All of the simulations presented here were use runs of $10^6$ time steps of $0.001$ seconds, and a coherence time for each of the exogenous random time series of 1 second. For two such series, independently generated, we find a standard deviation of the correlation (estimated from 1000 runs) of 0.023, compared with 0.001 for the same quantity of independent samples of white noise.  Therefore, for the examples reported here, any correlation whose magnitude is below about 0.05 must be judged statistically not significant.\footnote{Significance, in the everyday sense of practical usefulness, deserves a mention. For the practical task of estimating the value of one variable from another it is correlated with, far higher correlations are required. For a bivariate Gaussian, even with a correlation of 0.5, the probability of guessing from the value of one variable just whether the other is above or below its mean is only $67\%$, compared with $50\%$ for blind guessing. To be right 9 times out of 10 requires a correlation of 0.95. To estimate its value within half a decile 9 times out of 10 takes a correlation of 0.995. And to be sure from a finite sample that the correlation really is that high would require its true value to be even higher.}
The correlations observed here between $D$ and both $O$ and $P$ are thus indistinguishable from zero. These correlations are approximately summarised in Table~\ref{table-ex1-integrating-rounded}.
\begin{table}
\[\begin{array}{c|cc}
         &        P &        D \\ \hline
       O &        0 &       -1 \\
       P &          &        0 \\
\end{array}\]
\caption{Rounded correlations (Example 1).}\label{table-ex1-integrating-rounded}
\end{table}
This behaviour is quite different from that of a passive equilibrium system, such as a ball in a bowl (or something nailed down, which is a similar situation with a much higher spring constant). In the latter system, if we identify $D$ with an external force applied to the ball, $P$ with its position, and $O$ with the component of gravitational force parallel to the surface of the bowl, we will find (assuming some amount of viscous friction, and a measurement timescale long enough for the system to always be in equilibrium), that $O$ and $P$ are both proportional to $D$. There will also be a steady-state error.  This is not the case for the control system above, which has zero steady-state error. Given any constant value of $D$, $O$ will approach a value proportional to $D$ while $P$ tends to zero with time.

Real control systems often have to deal with some amount of transport lag in the environment, which we can model by changing Equation~\ref{eqn-environment} to $P(t) = O(t-\lambda) + D(t)$, where $\lambda$ is the amount of time delay. Transport lags are common in control systems in which the environment literally transports a substance $O$ from where the controller produces it to where it affects $P$. Examples abound in chemical process engineering and in biological systems.  This particular control system will only be stable in the presence of lag if its gain is below about $1/\lambda$.  When this is so, the correlations and rejection ratio are little affected by the presence of lag.  This remains true if correlations are calculated between lagged and unlagged quantities.

\subsection*{Example 2}

If we modify Example 1 by letting $R$ vary in the same manner as $D$, but independently from it, the correlations are now as shown in Tables~\ref{table-ex2-integrating} and~\ref{table-ex2-integrating-rounded}.
\begin{table}
\[\begin{array}{c|cccc}
         &        P &        R &        E &        D \\ \hline
       O &    0.718 &    0.717 &   -0.002 &   -0.718 \\
       P &          &    0.998 &   -0.027 &   -0.031 \\
       R &          &          &   \hfill 0.039 &   -0.032 \\
       E &          &          &          &   -0.024 \\
\end{array}\]
\caption{Correlations for varying $R$ and $D$ (Example 2).}\label{table-ex2-integrating}
\end{table}
\begin{table}
\[\begin{array}{c|cccc}
        &            P &            R &       E &          D \\ \hline
       O&          0.7 &          0.7 &       0 &       -0.7 \\
       P&              &            1 &       0 &          0 \\
       R&              &              &       0 &          0 \\
       E&              &              &         &          0
\end{array}\]
\caption{Rounded correlations (Example 2).}\label{table-ex2-integrating-rounded}
\end{table}
$O$ now has a substantial correlation with every variable except the only one that it directly depends on, $E$, with which its correlation is zero.  In the limit of increasing gain, $P$ and $O$ are almost identical to $R$ and $R-D$ respectively.  Since $R$ and $D$ are independently identically distributed, the correlation of $O$ with $R$ or $D$ tends to $\pm 1/\sqrt{2}$. All but one of the causal links (that from $O$ to $P$) has a correlation of zero; all but one of the non-zero correlations corresponds to a causal link.
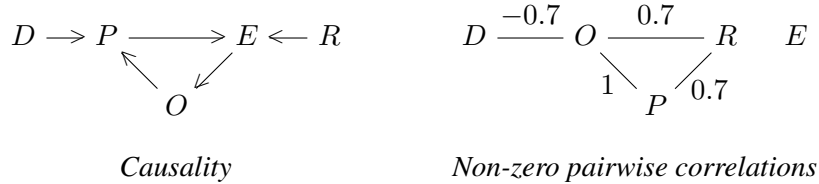
\[
\begin{array}{ccc}
\begin{diagram}[size=1.2em,heads=littlevee]
D & \rTo & P &       & \rTo &       & E & \lTo & R \\
  &      &   & \luTo &      & \ldTo \\
  &      &   &       & O 
\end{diagram}
& &
\begin{diagram}[size=1.2em,nohug]
D & \rLine^{-0.7} & O &              & \rLine^{0.7}  &               & R & & E \\
  &               &   & \rdLine_{1}  &               & \ruLine_{0.7} \\
  &               &   &              & P
\end{diagram} \\
\\
\mbox{\textit{Causality}} & \hspace{2em} & \mbox{\textit{Non-zero pairwise correlations}}
\end{array}
\]

\subsection*{Example 3}

In all of the systems discussed so far, there has been no noise---that is, signals of which the experimenter knows nothing except possibly their general statistical characteristics. Some of the variables are randomly generated waveforms, but they are all measured precisely, with no exogenous noise variables.  In the next example, we show that the introduction of modest amounts of noise can destroy some of the correlations among variables whose amplitudes are far larger than the noise.

In Example~1, if we measure the correlation between $O+D$ and $P$, then it will of course be identically 1, and we might consider this correlation to be important.
However, in practice, while the variables $P$, $R$, and $O$ may be accurately measureable, $D$ in general is not: it represents all the other influences on $P$ of any sort whatever, known or unknown. (The control system itself---Equation~\ref{eqn-controller}---does not use the value of $D$ at all. It senses only $P$ and $R$, and controls $P$ without knowing any of the influences on $P$.) To model our partial ignorance concerning $D$, we shall split it into $D_0$, the disturbances that can be practically measured, and $D_1$, the remainder. Let us assume that $D_0$ and $D_1$ are independently randomly distributed, and that the variance of $D = D_0 + D_1$ is ten times that of $D_1$. That is, $90\%$ of the variation of the disturbance is accounted for by the observed disturbances.
The correlations that result in a typical simulation run, with randomly varying $D$ and constant $R$ are listed in Tables~\ref{table-ex3-noisy} and~\ref{table-ex3-noisy-rounded}.

\begin{table}
\[\begin{array}{c|ccccc}
         &        P &    O+D_0 &          D_0 &          D_1 &            D \\ \hline
       O &    0.002 &    0.308 &       -0.947 &       -0.311 &       -0.999 \\
       P &          &    0.132 & \hfill 0.042 & \hfill 0.011 & \hfill 0.043 \\
   O+D_0 &          &          & \hfill 0.012 &       -0.990 &       -0.302 \\
     D_0 &          &          &              &       -0.006 & \hfill 0.948 \\
     D_1 &          &          &              &              & \hfill 0.311 \\
\end{array}\]
\caption{Correlations with noisy measurement (Example 3).}\label{table-ex3-noisy}
\end{table}
\begin{table}
\[\begin{array}{c|ccccc}
        &       P &   O+D_0 &     D_0 &      D_1 &      D \\ \hline
       O&       0 &   \weak &      -1 &   -\weak &      1 \\
       P&         &  \vweak &       0 &        0 &      0 \\
   O+D_0&         &         &       0 &       -1 & -\weak \\
     D_0&         &         &         &        0 &      1 \\
     D_1&         &         &         &          &  \weak
\end{array}\]
\caption{Rounded correlations (Example 3).}\label{table-ex3-noisy-rounded}
\end{table}

When $D_1$ has amplitude zero, the system is identical to the earlier one, for which $O+D_0 = P$.
But when the additional disturbance $D_1$ is added, accounting for only one tenth the variation of $D$,
the correlation between $O+D_0$ and $P$ sinks to a low level. The reason is that the variations in $P$ are much smaller than the noise we have introduced. For this run, the standard deviations were $\sigma(O) = 0.999$, $\sigma(D_0) = 0.953$, $\sigma(D_1) = 0.318$, and $\sigma(P) = 0.046$.
So although the unmeasurable $O+D_0+D_1$ is identical to $P$, the measurable $O+D_0$ correlates only weakly with $P$, and the better the controller controls, the smaller the correlation.

\subsection*{A digression on disturbances}

Example~3 also demonstrates the error of a common naive idea about control systems.
In every functioning control system, $O$ is just the output required to oppose the effects of the disturbances $D$ on $P$. It is sometimes assumed that a control system works by sensing $D$ and calculating the value of $O$ required to offset it. This does not work. It would result in every unsensed disturbance affecting $P$ unopposed, and as we have seen, the controller in Example~3 performs far better than this. It would also require the control system to contain a detailed model of how the output and the disturbances affect $P$---that is, a model of its environment. Any inaccuracies in this model will also produce unopposed influences on $P$.  For Example~4 below, control by calculating $O$ from measurements of $D_O$ and $D_P$ would allow any error in measuring $D_O$ to produce an error in $P$ growing linearly with time, a complete failure of control.

While controllers can be designed that do make some use of an environmental model and sensed disturbances to improve their performance, it remains the case that no control is possible without also sensing the actual variable to be controlled.  Control cannot be any better than the accuracy of that measurement. Sensing or modelling anything else is not a necessity for control.

\subsection*{Example 4}

A slightly different control system is illustrated in Figure~\ref{figure-proportionalcontroller}.
\begin{figure}
\centering
\resizebox{7.15cm}{!}{\includegraphics{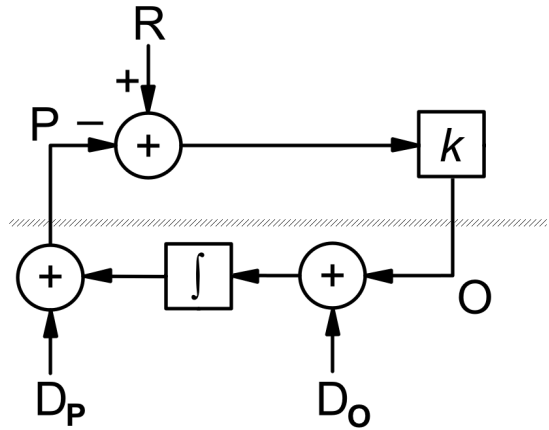}}
\caption{Another simple feedback control system.}\label{figure-proportionalcontroller}
\end{figure}
This is very similar to the previous one, but now the output is proportional to the error and the integrator is part of the environment. Two disturbances are present, $D_O$ adding to the output and $D_P$ adding to the perception. $D_P$ has standard deviation 1, and $D_O$ has a standard deviation chosen to produce the same magnitude of error signal as $D_P$.
These are the equations describing the system:
\[\begin{array}{rclr}
O & = & k(R - P) & \mbox{\hspace{2em}\textit{controller}}\\
P & = & D_P + \int (O + D_O)\,dt & \mbox{\hspace{2em}\textit{environment}}
\end{array}\]
With a constant reference signal, the correlations we obtain for a sample run of this system are shown in Table~\ref{table-ex4a}, and in rounded form in Table~\ref{table-ex4b}.
\begin{table}
\[\begin{array}{c|cccccc}
         &        P &        R &        E &      D_O &      D_P \\ \hline
       O &   -0.041 &    0.034 & \hfill 1.000 &   -0.527 &   -0.031 \\
       P &          &    0.997 &   -0.041 & \hfill 0.008 & \hfill 0.001 \\
       R &          &          & \hfill 0.034 &   -0.032 &   -0.002 \\
       E &          &          &          &   -0.527 &   -0.031 \\
     D_O &          &          &          &          &   -0.006 \\
\end{array}\]\caption{Correlations for a proportional controller (Example 4).}\label{table-ex4a}
\end{table}
The general pattern of correlations is this:
\begin{table}
\[\begin{array}{c|cccccc}
         &        P &        R &        E &       D_O &      D_P \\ \hline
       O &        0 &        0 &        1 &      -0.5 &        0 \\
       P &          &        1 &        0 &         0 &        0 \\
       R &          &          &        0 &         0 &        0 \\
       E &          &          &          &      -0.5 &        0 \\
     D_O &          &          &          &           &        0 \\
\end{array}\]\caption{Rounded correlations (Example 4).}\label{table-ex4b}
\end{table}
Again we see patterns of correlation that do not resemble the causal relationships.
$O$ is proportional to the difference between $P$ and $R$, and is not directly influenced by anything else, but is correlated with neither of them. $D_P$ has no correlation with any other variable. $O$ correlates perfectly with $E$, because $O = k E$ by definition, but the only other variable it has a correlation with is $D_O$, whose causal influence on $O$ proceeds via an almost complete circuit of the control loop.
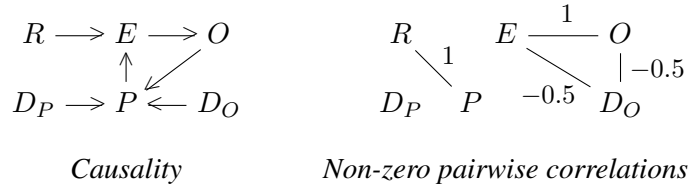
\[
\begin{array}{ccc}
\begin{diagram}[size=1.2em,heads=littlevee]
R   & \rTo & E    & \rTo  & O \\
    &      & \uTo & \ldTo \\
D_P & \rTo & P    & \lTo  & D_O 
\end{diagram}
& &
\begin{diagram}[size=1.2em,nohug]
R   &         & & E    & \rLine^{\mbox{\hspace{1em}\small $1$}}  & O \\
    & \rdLine^{\mbox{\small $1$}} & &      & \rdLine_{\mbox{\small $-0.5$}} & \dLine_{\mbox{\small $-0.5$}} \\
D_P &         & P    & &         & D_O 
\end{diagram} \\
\\
\mbox{\textit{Causality}} & \hspace{1em} & \mbox{\textit{Non-zero pairwise correlations}}
\end{array}
\]
As with the integral control system, the addition of transport lag of no more than $1/k$ does not change this behaviour.

The apparently paradoxical relationships in Examples~1--4 between causation and correlation can in fact be used to discover the presence of control systems. Suppose that $D_0$ is an observable and experimentally manipulable disturbance which one expects to affect a variable $P$ (because one can see a physical mechanism whereby it should do so), but when $D_0$ is experimentally varied, it is found that $P$ remains constant, or varies far less than expected. This should immediately suggest the presence of a control system which is maintaining $P$ at a reference level. This is called the Test for the Controlled Variable by \cite{Powers74,Powers98}. Something else must be happening to counteract the effect of $D_0$, and if one finds such a variable $O$ that varies in such a way as to do this, then one may hypothesise that $O$ is the output variable of the control system. Further investigation would be required to discover the whole control system: the actual variable being perceived (which might not be what the experimenter is measuring as $P$, but something closely connected with it), the reference $R$ (which is internal to the control system), the actual output (which might not be exactly what the experimenter is observing as $O$), and the mechanism that produces that output from the perception and the reference.
For this test to be effective, the disturbance must be of a magnitude and speed that the control system can handle. Only by letting the control system control can one observe that it is doing so.

A further paradox of control systems also appears here. Unlike the control systems of the earlier examples, this one has a certain amount of steady-state error.  When $R$, $D_O$, and $D_P$ are constant, the steady-state values of $E$ and $O$ are $-D_O/k$ and $-D_O$ respectively.  When the gain $k$ is large, $E$ may be experimentally indistinguishable from zero, while $O$ may be clearly observed.  $E$ is, however, the only immediate causal ancestor of $O$.  For the systems of Examples~1--3, the paradox is even stronger: the steady-state error and output signals are both zero, for any constant values of the disturbances and reference.  Superficial examination of the system could lead one to wrongly conclude that they play no role in its functioning.

\subsection*{Example 5}

In all of the above examples, the disturbances and the reference signal vary only on timescales much longer than the settling time of the control system: 1 second vs.~0.01 seconds.  If we repeat these simulations with signals varying on a much shorter timescale than the settling time, then we find completely different patterns of correlation.
Table~\ref{table-ex5} shows the general form of the results for examples 1--4 with a gain of 1 (giving a settling time of 1 second) and a coherence time for the random signals of 0.1 seconds.

\begin{table}
\begin{center}
\begin{tabular}{ccc}
$\begin{array}{c|cc}
         &        P &        D \\ \hline
       O &        0 &        0 \\
       P &          &        1 \\
\end{array}$
&&
$\begin{array}{c|ccccc}
        &       P &   O+D_0 &     D_0 &      D_1 &      D \\ \hline
       O&       0 &       0 &       0 &        0 &      0 \\
       P&         &       1 &       1 &    \weak &      1 \\
   O+D_0&         &         &       1 &        0 &      1 \\
     D_0&         &         &         &        0 &      1 \\
     D_1&         &         &         &          &      0
\end{array}$
\\
Example 5(1) & & Example 5(3)
\\
\\
$\begin{array}{c|cccc}
        &            P &            R &           E &          D \\ \hline
       O&            0 &            0 &           0 &          0 \\
       P&              &            0 &        -0.7 &          1 \\
       R&              &              & \hfill  0.7 &          0 \\
       E&              &              &             &       -0.7
\end{array}$
&\hspace{2em}&
$\begin{array}{c|cccccc}
         &        P &        R &          E &       D_O &        D_P \\ \hline
       O &     -0.7 &      0.7 &          1 &         0 &       -0.7 \\
       P &          & 0 \hfill &       -0.7 &         0 & \hfill 0.7 \\
       R &          &          & \hfill 0.7 &         0 &          0 \\
       E &          &          &            &         0 &       -0.7 \\
     D_O &          &          &            &           &          0 \\
\end{array}$
\\
Example 5(2) & & Example 5(4)
\end{tabular}
\end{center}\caption{Correlations in the presence of fast disturbances (Example 5).}\label{table-ex5}
\end{table}

Notice that Faithfulness is still violated.  For example, the new version of Example~1 has a zero correlation between $O$ and $P$, although the only causal arcs involving $O$ are with $P$. For this system, $O$ is just the integral of $P$, and so Theorem~\ref{theorem-zerocorr-realline} applies: $O$ and $P$ necessarily have zero correlation regardless of the signal injected as $D$. For the new versions of Examples~2 and~3, $O$ does not correlate with any other variable, and for Example~4, $D_O$ correlates with no other variable.  When a variable correlates with nothing else, no system of causal inference from non-interventional data can reveal a causal connection that includes it.

\subsection*{Summary}

We could go through many more examples of control systems of different architectures and different types of disturbances, but these should be enough to illustrate the general pattern.  The presence of control systems typically results in patterns of correlation among the observable variables that bear no resemblance to their causal relationships.  These patterns are robust: neither varying the parameters of the control systems nor collecting longer runs of data would reveal stronger correlations or diminish the extreme correlations between variables with no direct connection.

We shall later discuss why this is so, but first we shall examine some current work on causal inference
and demonstrate in each case that the authors' hypotheses exclude systems such as the above, and, if applied despite that, that their methods indeed fail to correctly discover the true causal relations.

\section{Current causal discovery methods applied to these examples}\label{section-literature}

When a theorem shows that certain causal information can be obtained from non-experimental observations on some class of systems, and yet no such information is obtainable from the systems we have exhibited, it follows that these systems must lie outside the scope of these theorems.  Here we shall survey some such methods and exhibit where the systems we consider here violate their hypotheses.

Dynamical systems are excluded from all of the causal analyses of both \cite{Pearl00} and~\cite{GSS01}, which do not consider time dependency.  In addition, control systems inherently include cyclic dependencies: the output affects the perception and the perception affects the output.  Control systems therefore fall outside the scope of any method of causal analysis that excludes cycles, and both of these works restrict attention to directed acyclic causal graphs.

\cite{LacerdaSRH08} consider dynamical systems sampled at intervals, and allows for cyclic dependencies, but a condition is imposed that in any equation giving $x_{n+1}$ as a weighted sum of the variables, possibly including $x$, at time $n$, the coefficient of $x_n$ in that sum must be less than 1. This excludes any relation of the form $x = dy/dt$, which in discrete time is approximated by the difference equation $x_{n+1} = x_n + y_n\,\delta t$.
In addition, they recommend sampling such systems on timescales longer than any transient effects. As can be seen from Figure~\ref{figure-VI}(c,d,e), the organised trajectories visible when the system is sampled on a short time scale vanish at longer sampling intervals: only the transient effects reveal anything about the true relation between the variables.
This recommendation thus \textit{rules out} any possibility of discerning causal influences from nonexperimental data in the presence of control systems.

The problems that dynamical systems pose for causal analysis have been considered in terms of the concept of ``equilibration''.
\cite{Iwasaki-1994-CausalityModelAbstraction} demonstrate how the apparent causal relationships in a dynamical system may depend on the timescale at which one views it, the timescale determining which of the feedback loops within the system have equilibrated. (Although he describes a variable which depends on its own derivative as ``self-regulating'', the didactic example that he discusses, of a leaky bathtub, is not a control system.)

\cite{Dash-2003-Caveats,Dash05} considers the interaction of the \textit{Equilibration} operator of \cite{Iwasaki-1994-CausalityModelAbstraction} and the \textit{Do} operator of \cite{Pearl00}, and considers the question of when they commute. He shows that they often do not, and recommends that when this is so, manipulation be performed before equilibration. This amounts to recommending that the system be acted on and sampled on a timescale shorter than its equilibration time, so that transient behaviour may be observed.  Even when this is done, however, the true causal relationships may fail to manifest themselves in the correlations, as shown by the examples of section~\ref{section-zeroderiv}, and in Example~5, where the disturbance is ten times as fast as the controller's response time.
When intermediate timescales of disturbance are applied to the example control systems, some mixture of the equilibrium and non-equilibrium patterns of correlation will be seen.
With more complicated systems of several control systems acting together at different timescales (as in the case of cascade control, where the output of one control system sets the reference of another, usually faster one), the patterns of correlations in the face of rapid disturbances will be merely confusing.
A further moral to be drawn from this is that ``equilibration'' is not necessarily a passive process, like the ball-in-a-bowl described after Example~1, or the bathtub example studied in \cite{Iwasaki-1994-CausalityModelAbstraction,Dash-2003-Caveats,Dash05}. Unlike those systems, control systems typically exhibit very small or zero steady-state error (which is almost their defining characteristic).

\cite{Duvenaud08} are concerned with techniques for making successful predictions rather than learning the correct causal structure.  However, for some of our examples this is impossible even in principle. If a variable has zero correlation with any other variable, conditional on any set of variables, then no information about its value can be obtained from the other variables.  This is the case, for example, with Example~5(1) and the variable $O$, and for the examples of section~\ref{section-zeroderiv}.

\cite{Itani08} propose a method of causal analysis capable of deriving cyclic models.
They first generalise causal graphs to the cyclic case.  Whereas in the acyclic case one can demonstrate that the conditional probability distribution on the whole graph factorises into conditional distributions of each variable given its immediate causes, this is not always so for cyclic graphs.
They therefore consider only graphs for which a generalisation of this holds.
Given the conditional distribution of each variable given its immediate causal predecessors,
a joint distribution of all the variables is said to be \textit{induced} by them if (i) a local version of factorisation holds, and (ii) nodes are independent under $d$-separation. (We refer to that paper for the complete technical definition.)  The requirement that the authors impose is that the conditional distributions must induce a unique global distribution.

This fails for the control systems considered here, because the presence of high correlations between variables connected by paths of low correlation results in non-uniqueness of such a global distribution.
For Example~1, the causal graph is
$D \rightarrow P \;\raisebox{-1mm}{$\stackrel{\textstyle\leftarrow}\rightarrow$}\; O$.
The required condition involves distributions $f(D)$, $f(O;P)$, and $f(P;O,D)$.
The actual global distribution has $O$ and $D$ normally distributed with means of 0, standard deviations of 1, and correlation $-0.997$.  $P$ has distribution $\delta_P(O+D)$, by which we denote the distribution in which all of the probability is concentrated at $P=O+D$.
This implies what the conditional distributions must be: $f(D) = N(0,1)$, $f(O;P) = N(0,1)$, and $f(P;O,D) = \delta_P(O+D)$.
But these are induced by any global distribution of the form $P(O,D) \delta_P(O+D)$, where $P(O,D)$ is a bivariate Gaussian with any correlation, and unconditional standard deviations for $O$ and $D$ of 1.
This happens because the conditional distributions omit any information about the joint distribution of $O$ and $D$, variables which are connected only via a third variable with which they have no correlation.

\cite{Voortman-2010-DifferenceBasedCausality} consider dynamical systems with cyclic dependencies, but their results on the learnability of such systems depend on the Faithfulness assumption, which, the authors note, is violated when there is an equilibrium. Our Examples~1--4 all maintain equilibria, and indeed the authors' Theorem~2 fails to apply to them. However, the examples of section~\ref{section-zeroderiv}, and Example~5 are not systems in equilibrium.  Even those systems violate Faithfulness, and the conclusions of the authors' theorems do not hold for them.

For the dynamical systems of section~\ref{section-zeroderiv}, the property is trivially satisfied, but the same distribution between the two variables---a bivariate Gaussian with zero correlation---would satisfy the assumptions for the causal graph on $V$ and $I$ with no edges.  So long as data are collected at intervals long compared with the coherence time of the waveforms, none of the four possible graphs on two nodes can be excluded.

\cite{zhang2008detection} consider the possibility of making a weaker Faithfulness assumption which is still sufficient to conduct causal analysis.
They demonstrate that Faithfulness implies two properties which they call
\textit{Adjacency Faithfulness} and \textit{Orientation Faithfulness}
and that, while these do not together imply Faithfulness, they are (given the Markov assumption)
all that is necessary for standard methods of causal inference.
They then prove that if Adjacency Faithfulness is satisfied, then Orientation Faithfulness can be tested from the data, obviating the need to assume it.  If Orientation Faithfulness fails the test, then the data are not faithful to any causal graph.
They also find a condition weaker than Adjacency Faithfulness, called \textit{Triangle Faithfulness}, having a similar property: if Triangle Faithfulness is satisfied, Adjacency Faithfulness can be tested from the data. Each of these Faithfulness conditions is a requirement that some class of conditional correlations be non-zero.

The work is restricted to acyclic graphs, so control systems are ruled out on that ground.
However, if we investigate Triangle Faithfulness anyway for our examples,
we find that the only triangle present in any of our causal graphs is the cyclic triangle connecting $P$, $E$, and $O$ in Examples~2 and~4. The three vertexes of this triangle are all non-colliders (i.e.~the causal arrows do not both go towards the vertex).  Triangle Faithfulness requires all of the correlations between any two of these vertexes, conditional on any set of variables not including the third, to be non-zero.
However, in Example~2, six of these twelve correlations are zero. The following table shows the correlations obtained from simulation (round values) between each pair of variables named at the left, conditional on the set of variables shown in the top row.
\[\begin{array}{c|cccc}
   & \mbox{none} &  R & D & RD \\ \hline
OP & 0.7         &  0 & 1 &  1 \\
OE & 0           &  0 & 0 & -1 \\
PE & 0           & -1 & 0 & -1
\end{array}\]
The values that are here $\pm 1$ all result from mathematical identities. For example, $c(OP|D)=1$ because $P = O+D$.
They are nonetheless valid correlations.
In practice, measurement noise would make these correlations slightly smaller, but they will still be extreme.
Correlations more typical of real experiments can be obtained only by assuming gross amounts of measurement error.

In view of the cyclicity of the triangle, we might instead test the Triangle Faithfulness condition that applies in the case of a collider, but we fare no better.
Each of the following correlations would be required to be non-zero.
\newcommand{\undef}{\mbox{\textit{undef}}}
\[\begin{array}{c|cccc}
     & \mbox{none} & R      & D      & RD            \\ \hline
OP|E & 0.7         & \undef(0) & 1      & \undef(0.6) \\
OE|P & 0           & \undef(0) & \undef(0) & \undef(-0.3) \\
PE|O & 0           & -1     & 0      & \undef(-0.6)
\end{array}\]
As before, the $\pm 1$ entries are mathematical identities.
The undefined entries are due to the fact that fixing some of these variables may also fix one or both of the variables whose correlation is being measured.  For example, $c(OP|ER)$ is undefined because fixing $E$ and $R$ fixes $P = R - E$.
Adding measurement noise can make these correlations well-defined (indicated by the parenthetical values), but this creates several more zeroes.
Example~4 gives similar results.

Triangle Faithfulness is therefore not satisfied on any interpretation of how one might apply it to these graphs, and so the failure of Faithfulness for these examples is not detectable by this method.
No approach along these lines can avail for these systems,
because the data generated by each of them are in fact faithful to some causal graph---but in no case are the data faithful to the real causal graph.
%

\section{The fundamental problem}

We have seen that control systems display a systematic tendency to violate Faithfulness, whether they are at equilibrium or not.  Low correlations can be found where there are direct causal effects, and high correlations between variables that are only indirectly causally connected, by paths along which every step shows low correlation.
This follows from the basic nature of what a control system does: vary its output to keep its perception equal to its reference.  The output automatically takes whatever value it needs to, to prevent the disturbances from affecting the perception.  The very function of a control system is to actively destroy the data that current techniques of causal analysis work from.

What every controller does is hold $P$ close to $R$, creating a very strong statistical connection via an indirect causal connection. For constant $R$, variations in $P$ measure the imperfection of the controller's performance---the degree to which it is not doing what it is supposed to be doing.  This may be useful information if one already knows what it is supposed to be doing, as will typically be the case when studying an artificial control system, of a known design and made for a known purpose. However, when studying a biological, psychological, or social system which might contain control systems that one is not yet aware of, correlations between perception and action---in other terminology, input and output, or stimulus and response---must fail to yield any knowledge about how it works.
Current causal analysis methods can only proceed by making assumptions to rule out the problematic cases.  However, these problematic cases are not few, extreme, or unusual, but are ubiquitous in the life and social sciences. For psychology, this point has been made experimentally in \cite{Powers-1978-Spadework,Marken-2011-CausalityCorrelation}.

Control systems also create problems for interventions.  To intervene to set a variable to a given value, regardless of all other causal influences on it, one must act so as to block or overwhelm all those influences.  This is problematic when the variable to be intervened on happens to be the controlled perception of a control system.  One must either act so strongly as to defeat the control system's own efforts to control, or fail to successfully set the perception.  In the former case, the control system is driven into a state atypical of its normal operation, and the resulting observations may not be relevant to finding out how it works in normal circumstances.  In the latter case, all one has really done is to introduce another disturbing variable. One has not so much done surgery on the causal graph as attach a prosthesis.  Arguably, introducing a new variable is the only way of intervening in a system, in the absence of hypotheses about the causal relationships among its existing variables.

When interventions are performed, a lack of appreciation of the phenomena peculiar to control systems can lead to erroneous conclusions. A candle placed near to a room thermostat will not warm it up, but apparently cool the rest of the room. If one does not know the thermostat is there, this will be puzzling.  If one has discovered the furnace, and noticed that the presence of the candle correlates with reduced output from the furnace, one might be led to seek some mechanism whereby the furnace is sensing the presence of the candle.  And yet there is no such sensor: the thermostat senses nothing but its own temperature and the reference setting.  It cannot distinguish between a nearby candle, a crowd of people in the room, or warm weather outside.

To test for the presence of control systems, one must take a different approach, by applying disturbances and looking for variables that remain roughly constant despite there being a clear causal path from the disturbance to that variable (the Test for Controlled Variables). When both a causal path and an absence of causal effect are observed, it is evidence that a control system may be present. If, at the same time, something else changes in such a way as to oppose the disturbance, that is a candidate for the control system's output action.

Discovering exactly what the control system is perceiving, what reference it is comparing it with, and how it generates its output action, may be more difficult to discover. For example, it is easy to demonstrate that mammals control their deep body temperature, less easy to find the mechanism that they sense that temperature with.  The task is made more difficult by the fact that in a well-functioning control system, the error may be so small as to be practically unmeasurable, even though the error is what drives the output action.

\section{Conclusion}

Dynamical systems exhibiting equilibrium behaviour are already known to be problematic for causal inference, although methods have been developed to extend methods of causal inference to include some parts of this class. But the subclass of control systems poses fundamental difficulties which cannot be addressed by any extension along those lines. They specifically destroy the connections between correlation and causation which these methods depend on.

The investigations here have been theoretical.
It remains to be seen how substantially these phenomena affect causal analysis of the increasingly massive data sets that are being gathered from gene expression arrays and high-resolution neuroimaging techniques.

\appendix

\section{Sufficient conditions for zero correlation between a function and its derivative}

Here we demonstrate the absence of correlation between a function satisfying certain weak boundedness conditions and its first derivative.
Before attending to the technicalities, we note that the proofs for both theorems
are almost immediate from the observation that $\int_a^b x\,\dot{x}\,dt = [\frac{1}{2}x^2]_a^b$.

\newcommand{\xab}{\overline{x}_{a,b}}
\newcommand{\xdotab}{\overline{\dot{x}}_{a,b}}
\begin{theorem}\label{theorem-zerocorr-interval}
Let $x$ be a differentiable real function, defined in the interval $[a,b]$,
such that $x(a) = x(b)$.
If $x$ is not constant then the correlation of $x$ and $\dot{x}$ over $[a,b]$ is defined and equal to zero.
\end{theorem}
\begin{proof}
Write $\xab$ and $\xdotab$ for the means of $x$ and $\dot{x}$ over $[a,b]$.
By replacing $x$ by $x - \xab$ we may assume without loss of generality that $\xab$ is zero.
$\xdotab$ must exist and equal zero, since
\[
\xdotab \;\; = \;\; \frac{1}{b-a} \int_a^b \dot{x} \,dt
\;\; = \;\; \frac{x(b)-x(a)}{b-a}
\;\; = \;\; 0
\]
The correlation between $x$ and $\dot{x}$ over $[a,b]$ is defined by:
\begin{eqnarray*}
c_{x,\dot{x}} & = &
\frac{\frac{1}{b-a} \int_a^b x \, \dot{x} \, dt}%
{\sqrt{(\frac{1}{b-a} \int_a^b x^2 \, dt)
\;
(\frac{1}{b-a} \int_a^b \dot{x}^2 \, dt)}} \\
& = &
\frac{(x(b)^2 - x(a)^2)/2}%
{\sqrt{(\int_a^b x^2 \, dt)
\;
(\int_a^b \dot{x}^2 \, dt)}}
\end{eqnarray*}
The numerator is zero and the denominator is positive (since neither $x$ nor $\dot{x}$ is identically zero).
Therefore $c_{x,\dot{x}} = 0$.
\end{proof}

\begin{theorem}\label{theorem-zerocorr-realline}
Let $x$ be a differentiable real function.
Let $\overline{x}$ and $\overline{\dot{x}}$ be the averages of $x$ and $\dot{x}$ over the whole real line.
If these averages exist, and if the correlation of $x$ and $\dot{x}$ over the whole real line exists, then the correlation is zero.
\end{theorem}
\begin{proof}
Note that the existence of the correlation implies that x is not constant.
As before, we can take $\overline{x}$ to be zero and prove that $\overline{\dot{x}}$ is also zero.
The correlation between $x$ and $\dot{x}$ is then given by the limit:
\begin{eqnarray*}
c_{x,\dot{x}} & = & \limab
\frac{\frac{1}{b-a} \int_a^b x \, \dot{x} \, dt}%
{\sqrt{(\frac{1}{b-a} \int_a^b x^2 \, dt)
\;
(\frac{1}{b-a} \int_a^b \dot{x}^2 \, dt)}} \\
& = & \limab
\frac{(x(b)^2 - x(a)^2)/2}%
{\sqrt{(\int_a^b x^2 \, dt)
\;
(\int_a^b \dot{x}^2 \, dt)}}
\end{eqnarray*}
Since this limit is assumed to exist, to prove that it is zero it is sufficient to construct some particular sequence of values of $a$ and $b$ tending to $\pm \infty$, along which the limit is zero.

Either $x(b)$ tends to zero as $b \rightarrow \infty$, or (since $\overline{x} = 0$ and $x$ is continuous) there are arbitrarily large values of $b$ for which $x(b) = 0$. In either case, for any $\epsilon > 0$ there exist arbitrarily large values of $b$ such that $|x(b)| < \epsilon$.
Similarly, there exist arbitrarily large negative values $a$ such that $|x(a)| < \epsilon$.
For such $a$ and $b$, the numerator of the last expression for $c_{x,\dot{x}}$ is less than $\epsilon^2/2$. However, the denominator is positive and non-decreasing as $a \rightarrow -\infty$ and $b \rightarrow \infty$. The denominator is therefore bounded below for all large enough $a$ and $b$ by some positive value $\delta$.

If we take a sequence $\epsilon_n$ tending to zero, and for each $\epsilon_n$ take values $a_n$ and $b_n$ as described above, and such that $a_n \rightarrow -\infty$ and $b_n \rightarrow \infty$, then along this route to the limit, the corresponding approximant to the correlation is less than $\epsilon_n/\delta$. This sequence tends to zero, therefore the correlation is zero.
\end{proof}

The conditions that $x(a)=x(b)$ in the first theorem and the existence of $\overline{x}$ in the second are essential. If we take $x = e^t$, which violates both conditions, then $\dot{x}=x$ and the correlation is 1 over every finite time interval.
That $\overline{\dot{x}}$ and $c_{x,\dot{x}}$ exist is a technicality that rules out certain pathological cases such as $x = \sin(\log(1+|t|))$, which are unlikely to arise in any practical application.

We remark that although we do not require them here,
corresponding results hold for discrete time series,
for the same reason in its finite difference form: that
$(x_i+x_{i+1})(x_{i+1}-x_i) = x_{i+1}^2-x_i^2$.

\bibliography{causnoncorr-arXiv}

\end{document}